\documentclass[reqno,12pt]{amsart}
\usepackage{amsmath,amsthm,amscd,amssymb}
\usepackage{setspace}
\newcommand{\dist}{\text{\rm{dist}}}

\newcommand{\ess}{\text{\rm{ess}}}

\newcommand{\supp}{\text{\rm{supp}}}

\newcommand{\beq}{\begin{equation}}
\newcommand{\eeq}{\end{equation}}
\newcommand{\ba}{\begin{align}}
\newcommand{\ea}{\end{align}}

\allowdisplaybreaks
\numberwithin{equation}{section}

\newtheorem{theorem}{Theorem}[section]

\newtheorem{proposition}[theorem]{Proposition}
\newtheorem{lemma}[theorem]{Lemma}
\newtheorem{corollary}[theorem]{Corollary}
\newtheorem{definition}[theorem]{Definition}

\theoremstyle{remark}

\newcounter{smalllist}

\begin{document}
\bibliographystyle{plain}
\title[Universality limits]{Universality Limits of a Reproducing Kernel for a Half-Line Schr\"odinger Operator and Clock Behavior of Eigenvalues}
\author[Anna Maltsev]{Anna Maltsev}
\thanks{$^*$ Mathematics 253-37, California Institute of Technology, Pasadena, CA 91125.
E-mail: amaltsev@caltech.edu}
\date{\today}
\keywords{Schr\"odinger operators, eigenvalues in a box}
\subjclass[2000]{34L05; 34L20; 34L40; 41A99}

\maketitle

\begin{abstract}
We extend some recent results of Lubinsky, Levin, Simon, and Totik from measures with compact support to spectral measures of Schr\"odinger operators on the half-line. In particular, we define a reproducing kernel $S_L$ for Schr\"odinger operators and we use it to study the fine spacing of eigenvalues in a box of the half-line Schr\"odinger operator with perturbed periodic potential. We show that if solutions $u(\xi, x)$ are bounded in $x$ by $e^{\epsilon x}$ uniformly for $\xi$ near the spectrum in an average sense and the spectral measure is positive and absolutely continuous in a bounded interval $I$ in the interior of the spectrum with $\xi_0\in I$, then uniformly in $I$
$$\frac{S_L(\xi_0 + a/L, \xi_0 + b/L)}{S_L(\xi_0, \xi_0)} \rightarrow \frac{\sin(\pi\rho(\xi_0)(a - b))}{\pi\rho(\xi_0)(a - b)},$$  where $\rho(\xi)d\xi$ is the density of states.
We deduce that the eigenvalues near $\xi_0$ in a large box of size $L$ are spaced asymptotically as $\frac{1}{L\rho}$. We adapt the methods used to show similar results for orthogonal polynomials.

\end{abstract}

\section{Introduction}
In this paper we exploit the similarities between differential and difference equations to show a half-line Schr\"odinger operator analogue of recent results of Lubinsky, Levin, Simon, and Totik.
Let $d\eta=w(x)dx + d\eta_s$ be a probability measure supported on $[-1, 1]$. Let the polynomials $p_n$ be orthonormal with respect to the $L^2(d\eta)$ inner product. The Christoffel--Darboux kernel $K_n$, given by
\begin{equation}
K_n(x, y) = \sum_{k=0}^{n}p_k(x)p_k(y).
\end{equation}
(see for example \cite{stahl-totik}, \cite{bulk}, \cite{simon-kernel-review}), is characterized by the reproducing property, i.e. for all $k<n$,
\begin{equation}
p_k(y) = \int K_n(x, y)p_k(x)d\eta(x).
\end{equation}

The measure $d\eta$ on a compact set $\mathfrak{e}$ is regular \cite{stahl-totik} if for any $\epsilon>0$ there exist $\delta>0$ and a constant $C$ so that
\begin{equation}
\sup_{\dist(y, \mathfrak{e})\leq \delta}|p_n(y, d\eta)|\leq Ce^{\epsilon n}.
\end{equation}

Let $I\subset (-1, 1)$ be a closed interval and $d\eta$ is regular such that $\supp(d\mu_s)\cap I = \emptyset$ and $w$ is continuous and nonvanishing on $I$. Then Lubinsky \cite{bulk} shows that for $a, b\in \mathbb{R}$ and uniformly for $x_0\in I$
\begin{equation}\label{Lubinsky's result}
\lim_{n\rightarrow \infty}\frac{K_n(x_0+\frac{a}{n}, x_0+\frac{b}{n})}{K_n (x_0, x_0)}= \frac{\sin(\pi \rho_{[-1,1]}(x_0)(a -b))}{\pi \rho_{[-1,1]}(x_0)(a-b)},
\end{equation}
where $\rho_{[-1, 1]}(x_0) = (\pi\sqrt{1-x_0^2})^{-1}$ is the density of states for $[-1, 1]$.
This result is interesting for both the study of orthogonal polynomials and of random matrices. It relates a fundamental object to the sine kernel and implies that the left hand side of (\ref{Lubinsky's result}) only depends on the continuity and positivity of the measure $d\eta$ at $x_0$ and its essential support. Additionally, Levin-Lubinsky \cite{levin-lubinsky} obtain the asymptotic spacing of the zeros of orthogonal polynomials near $x_0$ from (\ref{Lubinsky's result}).

In this paper we provide definitions of a reproducing kernel $S_L$ and of regularity for half-line Schr\"odinger operators. We prove the analogous results for perturbed periodic half-line Schr\"odinger operators.

Let
\begin{equation}\label{A}
A\phi(x) = -\frac{d^2 \phi(x)}{dx^2} + V(x)\phi(x)
\end{equation}
be a Schr\"odinger operator on $L^2[0,\infty)$ with either Dirichlet or Neumann boundary condition at $x=0$. We assume throughout that $V$ is locally integrable and bounded from below. Let $u$, $y$ be the standard fundamental solutions of the eigenvalue equation of the operator $A$:
\begin{equation}
A\phi(\xi, x) = \xi \phi(\xi, x)
\end{equation}
with initial conditions
\begin{equation}
u(\xi, 0)= 1 = y'(\xi, 0), u'(\xi, 0) = 0 = y(\xi, 0).
\end{equation}
Throughout the paper, $u'$, $y'$ denote the derivative with respect to $x$ and $\mathfrak{e} = \sigma_{\ess}(A)$. Our results are valid for both Dirichlet and Neumann boundary conditions, but we only give the proofs for the Neumann case.
There is a shift of notation here, so $x$ in our setup is analogous to $n$ of the discrete case, and our $\xi$ is the analogue of $x$ of the discrete case.

We now define a reproducing kernel $S_L$ for Schr\"odinger operators.
\begin{definition}
Given a Schr\"odinger operator $A$ as in (\ref{A}) with the Neumann boundary condition we let the reproducing kernel be
\begin{equation}
S_L(\xi, \zeta) = \int_0^L u(\xi , t)u(\zeta , t)dt.
\end{equation}
\end{definition}
There exists a measure $d\mu$ (Theorem 2.2.3 of \cite{marchenko}; we change variables from Marchenko so that his $\sqrt{\xi}$ is our $\xi$) which makes the following two formulas hold for every function $f\in L^2[0, \infty)$:
\begin{align}\label{complete}
W(\zeta , f) &= \int_0^{\infty} f(x)u(\zeta , x)dx \\
f(x) &= \int W(\zeta , f)u(\zeta , x)d\mu(\zeta).
\end{align}
Here $d\mu$ is the spectral measure of the operator $A$ as in (\ref{A}).
We see that the reproducing property is satisfied with respect to $d\mu$:
\begin{equation}\label{reproducing}
u(\xi , x)\chi_{[0, L]}(x) =
\int S_L(\xi, \zeta) u(\zeta, x)d\mu(\zeta).
\end{equation}

We are primarily interested in the case where the potential $V = q + p$ where $p$ is periodic with period $P$ and continuous.
\begin{definition} \label{non-destructive zero-average}We call a perturbation $q$ \textbf{non-destructive} if it leaves the essential spectrum unchanged and \textbf{zero-average} if
\begin{equation}\label{perturbation}
\frac{1}{x}\int_0^x|q(t)|dt\rightarrow 0.
\end{equation}
\end{definition}
We assume throughout that the perturbation $q$ is a non-destructive zero-average perturbation (e.g. $q\rightarrow 0$ at $\infty$).

We can now state our main result:
\begin{theorem}\label{off-diagonal kernel}
Let $A = -\frac{d^2}{dx^2} + p(x)+q(x)$ with periodic and continuous $p$ and non-destructive zero-average $q$ and let $d\mu(\xi) = w(\xi)d\xi + d\mu_s$ be its spectral a measure. Let $I\subset \mathfrak{e}^{int}$ be a closed and bounded interval such that $w$ is continuous and non-zero on $I$ and $\supp(d\mu_s)\cap I = \emptyset$. Let $\xi_0\in I$ and $a, b, B \in \mathbb{R}$. Then uniformly in $I$ and $|a|, |b|<B$
\begin{equation}\label{universality limit of the kernel}
\frac{S_L(\xi_0 + a/L, \xi_0 + b/L)}{S_L(\xi_0, \xi_0)} \rightarrow \frac{\sin(\pi\rho(\xi_0)(a - b))}{\pi\rho(\xi_0)(a - b)},
\end{equation}
where $\rho(\xi)d\xi$ is the density of states.
\end{theorem}

Like in the discrete case, the asymptotic behavior of the kernel $S_L$ for the perturbed periodic operator $A$ depends on the density of states $\rho(\xi)d\xi$ of the periodic operator $A^{\#}$, defined, for example, in Berezin-Shubin (Section 2.3 of \cite{berezin-shubin}). The measure $\rho(\xi)d\xi$ is the same for Dirichlet and Neumann boundary conditions.

It is well known that $$K_n(x, y) = \frac{p_n(x)p_{n-1}(y) - p_{n-1}(x)p_n(y)}{x-y}.$$
This expression is called the Christoffel--Darboux formula, and we show its analogue in Section \ref{variation and cd formula} for $S_L$.
From (\ref{universality limit of the kernel}) and the Christoffel--Darboux formula (\ref{Christoffel-Darboux}) we deduce that the zeros of $u'(-, L)$, scaled by the density of states, will be asymptotically equally spaced, like the zeros of the sine function. We adapt the definition from \cite{clockbehavior}:
\begin{definition}
Fix $\xi^*$ in an interval $I,$ and number the zeros $\xi_N$ of $u'(-, L)$ with increasing positive integers to the right of $\xi^*$ and decreasing negative integers to the left so that $...<\xi_{-1}<\xi^* \leq \xi_0 < ...$. We say there is \textbf{strong clock behavior} of zeros of $u'$ at $\xi^*$ on an interval $I$ if the density of states $\rho(\xi)d\xi$ is continuous and nonvanishing on $I$ and for fixed $n$ \begin{equation}\label{strong clock behavior}
\lim_{L\rightarrow\infty} L|(\xi_n - \xi_{n+1})|\rho(\xi^*)= 1,
\end{equation} and we say there is \textbf{uniform clock behavior} on $I$ if the limit in (\ref{strong clock behavior}) is uniform on $I$ for fixed $n$.
\end{definition}
This nomenclature comes from the theory of orthogonal polynomials on the unit circle. There, when zeros of polynomials exhibit clock behavior, they do indeed look like marks on a clock. In Section \ref{section off-diagonal kernel}, we show
\begin{corollary}\label{clock behavior}
Let $A$, $\mathfrak{e}$, $I$, $\xi_0$ as in Theorem \ref{off-diagonal kernel}. Then there is uniform clock behavior of the zeros of $u'$ and $y$ on $I$.
\end{corollary}
It is well-known (e.g. Theorems 5.18, 5.20 of \cite{teschl-ode}) that the spacing of eigenvalues for functions on $[0, L]$ is the same as the spacing of zeros of $y(\xi, L)$ in $\xi$ in case of the Dirichlet boundary condition at 0 and $L$ and of $u'(\xi, L)$ in case of the Neumann boundary condition.

In our setup, we use the space
\begin{equation}
H_L = \left\{\pi: \pi(\xi)=\int_0^L f(x)\cos(\sqrt{\xi} x)dx, f\in L^2[0, \infty)\right\}
\end{equation}
as the analogue of the space of polynomials with degree less than or equal to $n$. Just like polynomial degree, when two functions with parameters $L$ and $N$ are multiplied, if the product is in $H_M$ for some $M$ then $M = L+N$, since multiplying exponentials adds their exponents. The orthogonal polynomials with degree smaller than or equal to $n$ are a basis for the space of polynomials with degree less than or equal to $n$. The analogous property of $H_L$ is
$$H_L = \{\pi: \pi(\xi)=\int_0^L f(x)u(\xi, x)dx, f\in L^2[0, \infty)\}.$$
This follows easily from Marchenko
(see (1.2.10), (1.2.10'') of \cite{marchenko}), which gives the existence of a continuous integral kernel $M$, such that
\begin{equation}
 \pi(\xi) = \int_0^L f(x) \left(u(\xi, x) + \int_0^x M(x, t)u(\xi, t)dt\right)dx.
\end{equation}
The space of polynomials of degree less than or equal to $n$ is usually considered with the $L_2(d\eta)$ inner product. Analogously, we give $H_L$ the following inner product:
\begin{equation}
\langle\pi_1, \pi_2\rangle = \int \pi_1(\zeta )\overline{\pi_2(\zeta )}d\mu(\zeta),
\end{equation}
where $d\mu$ is the spectral measure.

The minimizer of $\|\pi(y)\|_{L^2(d\eta)}$ over polynomials $\pi$ with $\deg \pi \leq n$ and $\pi(x) = 1$ is equal to $\frac{K_n(x, y)}{K_n(x, x)}$ and the minimum is equal to $K_n(x, x)^{-1}$. This property is called the \textbf{variational principle} and we show its analogue for $S_L$:
\begin{theorem}\label{variational principle}
If $\mu$ is an unnormalized spectral measure, then
\begin{equation}\label{minimum}
\min\{\|Q\|_{d\mu}: Q\in H_L, Q(\xi_0)= 1\} = S_L(\xi_0, \xi_0)^{-1},
\end{equation}
and the minimizer is given by
\begin{equation}\label{minimizer}
\frac{S_L(\xi, \xi_0)}{S_L(\xi_0, \xi_0)}.
\end{equation}
\end{theorem}
We give the minimum its own letter:
\begin{equation}\lambda_L(\xi) = S_L(\xi, \xi)^{-1}.
\end{equation}

Returning to the orthogonal polynomials case for motivation, we summarize Lubinsky's method for showing (\ref{Lubinsky's result}).
He notes that if $d\eta,$ $d\eta^{*}$ are regular measures on $[-1, 1]$ with $d\eta \leq d\eta^{*}$ and $K^*$ is the Christoffel-Darboux kernel associated with $d\eta^{*}$,
\begin{equation}\label{Lubinsky inequality}
\frac{|K_n(x, y)-K^*_n(x, y)|}{K_n(x, x)}\leq \left(\frac{K_n(y, y)}{K_n(x, x)}\right)^{1/2}\left(1-\frac{K^*_n(x, x)}{K_n(x, x)}\right)^{1/2}.
\end{equation}
This inequality, called \textbf{Lubinsky's inequality}, implies that in order to understand the left hand side of (\ref{Lubinsky's result}), it is sufficient to understand $K_n^{\#}(x, y)$ for some model measure $d\eta^{\#}$ and the behavior of a ratio of diagonal kernels. A model $d\eta^{\#}$ with $w^{\#}(x_0)=w(x_0)$ is chosen, for which $K_n^{\#}(x, y)$ can be computed directly. Then $d\eta^{*}=\sup\{d\eta^{\#}, d\eta\}$ dominates both $d\eta$ and $d\eta^{\#}$ and a has similarly nice local behavior at $x_0$ with $w^*(x_0) = w(x_0)$. By the variational principle, the ratios of the diagonal kernels $\frac{K^{\#}_n(x, x)}{K^*_n(x, x)}$ and $\frac{K_n(x, x)}{K^*_n(x, x)}$ both converge to 1, and Lubinsky's inequality and a comparison of the two resulting expressions yields the desired result. Simon \cite{simon2ext} and Totik \cite{totik} extend this argument to measures with $\supp_{\ess}(d\eta) = \cup I_j$ a finite union of intervals. In this paper we adapt all the steps to Schr\"odinger operators.

We adapt the regularity condition to spectral sets of half-line Schr\"odinger operators as follows:
\begin{definition}
Suppose $\mathfrak{e}\subset \mathbb{R}$ is the essential support of a spectral measure $d\mu$ of a Schr\"odinger operator with Neumann boundary condtion. We say $d\mu$ satisfies \textbf{regularity bounds} if for any $\epsilon >0$ there exists $\delta_1>0$, $C$ such that for all $\xi$ with $\dist(\xi, \mathfrak{e}) \leq \delta_1$ the solution $u$ satisfies
\begin{equation}\label{regularity bounds}
\int_0^L u(\xi, x)^2dx\leq Ce^{\epsilon L},
\end{equation}
with $C$ not dependent on $\xi$, $L$.
\end{definition}
In Section \ref{potential} we show that a Schr\"odinger operator with potential of the form $q(x) + p(x)$ with continuous periodic $p$ and non-destructive zero-average $q$ (as in Definition \ref{non-destructive zero-average}) satisfies regularity bounds.

Lubinsky's inequality carries over exactly to our setup, as we show in Section \ref{section off-diagonal kernel}. Similar to Simon \cite{simon2ext} and Lubinsky \cite{bulk}, we need a measure $d\mu^{\#}(\xi) = w^{\#}(\xi)d\xi + d\mu_s^{\#}$, which corresponds to a Schr\"odinger operator $A^{\#}$ and satisfies the following properties (we call such a measure a \textbf{model})

\begin{enumerate}

\item $\sigma_{\ess}(\mu^{\#}) = \mathfrak{e}$

\item $w^{\#}$ is continuous and nonvanishing on $\mathfrak{e}$

\item For any compact interval $I\subset \mathfrak{e}^{int}$ and $\epsilon > 0$ as $L\rightarrow\infty$ uniformly on $I$
\begin{equation}\label{model condition 3}
\sup_{\xi\in I}e^{-\epsilon L}S_L(\xi, \xi, d\mu^{\#})\rightarrow 0.
\end{equation}

\item For any compact interval $I\subset \mathfrak{e}^{int}$ for all $\xi\in I$ uniformly,
\begin{equation}\label{limsup}
\lim_{\epsilon \rightarrow 0}\lim_{L\rightarrow\infty}\frac{S_{L+\epsilon L}(\xi,\xi, \mu^{\#})}{S_L(\xi,\xi, \mu^{\#})} = 1.
\end{equation}

\item For $\xi(L)\rightarrow \xi_0$ in $\mathfrak{e}^{int}$
\begin{equation} \label{model condition 5}
\lim_{L\rightarrow\infty}\frac{S_L(\xi(L), \xi(L))}{S_L(\xi_0, \xi_0)}=1
\end{equation}
and this limit is uniform in $I$.
\end{enumerate}

We need these properties in the proof of Theorem \ref{diagonal kernel}. Theorem \ref{diagonal kernel calculation} immediately implies that the operator $A^{\#}$ with periodic potential satisfies model conditions 3-4. In Theorem \ref{theorem kernel calculation}, we notice that model condition 5 is satisfied. Thus, $A^{\#}$ is a model.
 We therefore can use the periodic potential as a model for $\mathfrak{e}$, whenever $q$ is non-destructive.

The essential spectrum $\mathfrak{e}$ of a periodic Schr\"odinger operator is a union of closed intervals. Let $\Delta$ be the discriminant of the periodic Schr\"odinger operator $A^{\#} = -\frac{d^2}{dx^2} + p$ (as in for example Chapter 2 of \cite{magnus-winkler}). The spectrum of $A^{\#}$ is the preimage of $[-2, 2]$ under $\Delta$. We let $\mathfrak{e} = \cup [l_n, r_n]$ so that $\Delta$ is a invertible on each $[l_n, r_n]$. We call each $[l_n, r_n]$ a band and each interval in $\mathbb{R}\backslash \mathfrak{e}$ a gap. When $r_n = l_{n+1}$, we call the point $\xi = r_n$ a closed gap.
The perturbed operator may have countably many eigenvalues in each gap, but the only limit points are the bands' endpoints.  Furthermore, there exists a first band, so shifting $q$ by a constant in energy, we can assume that $\min \mathfrak{e} = 0$. When $p$ is bounded, the size of the $n$th gap goes to 0 as $n\rightarrow\infty$ (Lemma 2.9 of \cite{magnus-winkler}), so only finitely many gaps and finitely many eigenvalues do not lie in $\{\xi:\dist(\xi, \mathfrak{e}) \leq \delta_1\}$ for any $\delta_1>0$.

The same is true for the comparison measure $d\mu^{*}$ which we construct to dominate both $d\mu$ and $d\mu^{\#}$ and to be continuous and non-vanishing on $I$ with $w^{*}(\xi_0) = w(\xi_0)$. We let $d\mu^{*}$ be the sup of $d\mu$, $d\mu^{\#}$ on a compact subset of $\mathbb{R}$ and $d\mu + d\mu^*$ on the rest of $\mathbb{R}$. The comparison measure is a scalar multiple of a spectral measure, as we show in Section \ref{section off-diagonal kernel}. We call such measures unnormalized spectral measures, as analogous to unnormalized measures on compact sets. If $u$, $y$ is a fundamental system of solutions and $S_L$ the reproducing kernel associated to a spectral measure $d\mu$, then  for $s>0$ we associate $\frac{u}{\sqrt{s}}$, $\frac{y}{\sqrt{s}}$, and the reproducing kernel $\frac{1}{s}S_L(\zeta, \xi, d\mu(\xi))$ to $d(s\mu)$. A spectral measure $d\mu$ must have a prescribed asymptotic at infinity (Theorem 2.4.2 of \cite{marchenko}), which implies that the normalization constant $s$ is unique and the reproducing kernel is well-defined. Henceforward, we use the letters $d\mu$, $d\mu^{*}$ to denote spectral measures which may be unnormalized and all results in Section \ref{variation and cd formula} are shown for unnormalized spectral measures. Also, the definition of regularity bounds works just as well.

In Section \ref{section bound}, we show

\begin{theorem}\label{diagonal kernel}
Suppose $d\mu(\xi) = w(\xi)d\xi + d\mu_s$, $d\mu^{*}(\xi) = w^{*}(\xi)d\xi + d\mu_s^{*}$ are unnormalized spectral measures with $\sigma_{\ess}(d\mu)=\sigma_{\ess}(d\mu^{*})=\mathfrak{e}$. Suppose $d\mu$, $d\mu^{*}$ satisfy regularity bounds and have finitely many eigenvalues outside of $\{\xi: \dist(\xi,\mathfrak{e}) < \delta_1\}$ for any $\delta_1>0$. Let $I\subset \mathfrak{e}^{int}$ be a closed and bounded interval such that $w, w^{*}$ are continuous and strictly positive on $I$ and $(\supp(d\mu_s)\cup \supp(d\mu^{*}_s))\cap I = \emptyset$. Let $\xi_0\in I$ and $\xi(L) \rightarrow \xi_0$ as $L\rightarrow\infty$. Then uniformly in $I$
\begin{equation}\label{equation diagonal kernel}
\frac{S_L(\xi(L), \xi(L), \mu)}{S_L(\xi(L), \xi(L), \mu^{*})}\rightarrow \frac{w^{*}(\xi_0)}{w(\xi_0)}.
\end{equation}
\end{theorem}

In Section \ref{model}, we compute the universality limit of the kernel in the unperturbed periodic case to be
\begin{equation}
\lim_{L\rightarrow\infty}\frac{S_L(\xi_0 +\frac{a}{L}, \xi_0 + \frac{b}{L})}{S_L(\xi_0, \xi_0)}= \frac{\sin(\pi\rho(\xi_0)(a - b))}{\pi\rho(\xi_0)(a - b)},
\end{equation}
where $\rho(\xi)d\xi$ is the density of states corresponding to the periodic Schr\"odinger operator. To make this calculation, we use a standard formula to express the density of states in terms of the imaginary part of the diagonal Green's function, and then we express the Green's function in terms of the solution $u$.

From Theorem \ref{diagonal kernel} and adapted Lubinsky's inequality we deduce Theorem \ref{off-diagonal kernel}.

As an example we consider the case $p =0$. In Section \ref{free} we show by direct computation that given same conditions on the measure as in Theorem \ref{off-diagonal kernel} we have
$$\lim_{L\rightarrow\infty}\frac{S_L(\xi + \frac{a}{L}, \xi+\frac{b}{L})}{S_L(\xi,\xi)}=\frac{\sin\left(\frac{a-b}{2\sqrt{\xi}}\right)(2\sqrt{\xi})}{a-b}$$
which yields that the eigenvalues in a box of size $L$ are spaced asymptotically as $\frac{1}{2L\sqrt{\xi}}$.

\section{The Perturbed Periodic Potential}\label{potential}

Let $\mathfrak{e}$ be the essential spectrum of a Schr\"odinger operator with period $P$ periodic potential $p$ and either Neumann or Dirichlet boundary condition.
 The goal of this section is to show
\begin{proposition}\label{regularity for perturbed periodic}
A Schr\"odinger operator with essential spectrum $\mathfrak{e}$ and potential $V(x) = p(x)+q(x)$ where $p$ is periodic and continuous and $\frac{1}{x}\int_0^x |q(t)|dt\rightarrow 0$ satisfies regularity bounds.
\end{proposition}

Fix $\epsilon>0$ and let $\frac{1}{x}\int_0^x |p(t)+q(t)| dt \leq M$ for $x > x_0$, some $x_0$. To prove (\ref{regularity bounds}), it is sufficient to show that $\int_0^Lu(\xi, x)^2dx\leq Ce^{\epsilon L}$ separately for three cases of $\xi$, where $C$ is uniform in $\xi$, $L$:
\begin{enumerate}
\item $\xi > \frac{4M^2}{\epsilon^2}$, shown in Lemma \ref{large xi}

\item $\xi \leq \frac{4M^2}{\epsilon^2}$, $\xi$ in the interior of $\mathfrak{e}$, but slightly away from the endpoints of the intervals, i.e. $\xi \in
\left(\cup [l_n + \epsilon, r_n - \epsilon]\right)\cap [0, \frac{4M^2}{\epsilon^2}]$, shown in Lemma \ref{interior}

\item $\xi \leq \frac{4M^2}{\epsilon^2}$ and $\xi$ near the interval endpoints i.e. $\xi \in (\cup [l_n - \epsilon, l_n + \epsilon] \cup [r_n - \epsilon, r_n + \epsilon])\cap [0, \frac{4M^2}{\epsilon^2}]$, shown in Lemma \ref{near endpoints}
\end{enumerate}

\begin{lemma}\label{large xi}
Let $A = -\frac{d^2}{dx^2} + V(x)$ be a Schr\"odinger operator with $V$ such that $\frac{1}{x}\int_0^x |V(t)|dt$ is bounded in $x$ as $x\rightarrow\infty$. Then the solutions $u$, $y$ of the eigenvalue equation satisfy
\begin{align}
u(\xi, x) &\leq Ce^{\frac{\int_0^x |V(t)|dt}{\sqrt{\xi}}}\label{bound for u in xi}\\
y(\xi, x) &\leq Ce^{\frac{\int_0^x |V(t)|dt}{\sqrt{\xi}}}.\label{bound for y in xi}
\end{align}
\end{lemma}
\begin{proof}
Using successive approximations, we can perturb about the solutions with $V=0$. Chadan-Sabatier ((I.2.3), (I.2.4), (I.2.6), (I.2.8a) \cite{chadan-sabatier}) show (\ref{bound for y in xi}), and using $\cos (\sqrt{\xi}x)$ as initial data instead of (I.2.3) gives (\ref{bound for u in xi}).
\end{proof}

This lemma indeed implies that for $\sqrt{\xi} \geq \frac{2M}{\epsilon}$ the solution $u$ satisfies $u(x)\leq Ce^{\frac{1}{2}\epsilon x}$, which implies $\int_0^Lu(\xi, x)^2dx\leq Ce^{\epsilon L}$.

\begin{lemma}\label{interior}
Let $[l_n, r_n]$ be a band of the spectrum for a Schr\"odinger operator $A = -\frac{d^2}{dx^2}+q(x)+p(x)$ with periodic and continuous $p$ and non-destructive zero-average $q$ (Definition \ref{non-destructive zero-average}). Then the solution $u$ of the eigenvalue equation with the Neumann boundary condition satisfies $\int_0^Lu(\xi, x)^2dx\leq Ce^{\epsilon L}$ for $\xi\in (\cup[l_n+\epsilon, r_n-\epsilon])\cap[0,R]$, where $R=\frac{4M^2}{\epsilon^2}$, and same holds for the solution with the Dirichlet boundary condition.
\end{lemma}

\begin{proof}

Let $u_p(\xi, x), y_p(\xi, x)$ be the solutions of $A^{\#} = -\frac{d^2}{dx^2}+p(x)$ with boundary conditions
\begin{align*}
u_p(\xi, 0)&= 1 = y_p'(\xi, 0)\\
y_p(\xi, 0) &= 0 = u_p'(\xi, 0)
\end{align*}

By Floquet's theorem (for example Section 1.2 of \cite{magnus-winkler} and Theorem XIII.89 of \cite{reed-simon4}), there exists a solution $f(\xi, x) = e^{i\theta(\xi)x}\phi(\xi, x)$, where $\phi$ is periodic in $x$ with period $P$. We normalize $f'(\xi, 0) =1$. The exponent $\theta(\xi)$ is not 0 or $\pi$ away from band endpoints, so that $\overline{f}$ is linearly independent of $f$ for $\xi \in \cup[l_n+\epsilon, r_n-\epsilon]$. Then
\begin{equation}
u(\xi, x) = a_1(\xi)f(\xi, x)+a_2(\xi)\overline{f(\xi, x)}.\end{equation}
We solve for $a_1, a_2$ in terms of $\xi$. We get that
\begin{align*}
1 &= u(\xi, 0) = a_1(\xi)f(\xi, 0) + a_2(\xi)\overline{f(\xi, 0)}\\
0 &= u'(\xi, 0) = a_1(\xi)f'(\xi, 0) + a_2(\xi)\overline{f'(\xi, 0)} = a_1(\xi) + a_2(\xi),
\end{align*}
so that $$a_1(\xi) = -a_2(\xi)$$
Substituting
$$1 = a_1(\xi)f(\xi, 0) - a_1(\xi)\overline{f(\xi, 0)},
$$
we get
$$a_1(\xi) = (2i\Im f(\xi, 0))^{-1} = -a_2(\xi).$$
Since $f$, $\overline{f}$ are independent, $\Im{f}\neq 0$ and, by Theorem XIII.89 of \cite{reed-simon4}, $f$ is analytic in $\xi$ on $[l_n+\epsilon, r_n-\epsilon]$. This implies that $a_1$, $a_2$ are analytic as well.
The function $|f|$ is continuous in both $x$ and $\xi$ on $[0,P]\times\left(\cup[l_n+\epsilon, r_n-\epsilon]\cap[0,R]\right)$, therefore it achieves its maximum on this set. Since $|f|$ is periodic and continuous in $x$ with period $P$, the maximum of $|f|$ in $x$ for fixed $\xi$ occurs on $[0,P]$. This implies that $u_p(\xi, x)\leq K$ where $K$ is constant in $x$ and $\xi\in \cup[l_n+\epsilon, r_n-\epsilon]\cap[0,R]$.

We use the method of variation of parameters about $u_p(\xi, -),$ $y_p(\xi, -)$ and Gronwall inequality. Let $a$, $b$ be given by
\begin{equation} \label{var of parameters}\left(\begin{array}{c}
u(x)\\
u'(x)\end{array}\right) =
\left(\begin{array}{cc}
u_p(x) & y_p(x)\\
u_p'(x) & y_p'(x)\end{array}\right)\left(\begin{array}{c}a(x) \\ b(x)\end{array}\right).\end{equation}
Then we get that
\begin{align*}
\left|\left(\begin{array}{c}
a(x)\\
b(x)\end{array}\right)\right| &= \left|\left(\begin{array}{c}
1\\
0\end{array}\right)+\int_0^xq(t)\left(\begin{array}{cc}
-y_p u_p & -y_p^2\\
u_p^2 & u_py_p \end{array}\right)\left(\begin{array}{c}a(t) \\ b(t)\end{array}\right)dt\right|\\
&\leq 1 + K_1\int_0^x|q(t)|\left|\left(\begin{array}{c}a(t) \\ b(t)\end{array}\right)\right|dt,
\end{align*}
where $K_1\geq |y_p u_p|+y_p^2+u_p^2+|u_py_p|$ is constant in $x$ and $\xi$ by the argument above.
We apply the Gronwall inequality to this integral equation to get
\begin{equation}
|a(x)|+|b(x)| \leq K_2e^{K_1\int_0^x|q(t)|dx}.
\end{equation}
Then we take the matrix norm in (\ref{var of parameters}) and, recalling that $\frac{1}{x}\int_0^x |q(t)|dt\rightarrow 0$, we get (\ref{regularity bounds}) for large $L$ and for all $L$ by choosing $C$ appropriately.
\end{proof}


\begin{lemma}\label{near endpoints}
Let $[l_n, r_n]$ be a band of the spectrum for a Schr\"odinger operator $A = -\frac{d^2}{dx^2}+q(x)+p(x)$ with continuous periodic $p$ and non-destructive zero-average $q$ (Definition \ref{non-destructive zero-average}). Then the solution $u$ of the eigenvalue equation with Neumann boundary condition satisfies $\int_0^Lu(\xi, x)^2dx\leq Ce^{\epsilon L}$ for $$\xi\in \left(\cup[l_n-\epsilon, l_n+\epsilon]\cup[r_n-\epsilon, r_n+\epsilon]\right)\cap\left[0,\frac{4M^2}{\epsilon^2}\right].$$ The same holds for the solution with the Dirichlet boundary condition.
\end{lemma}

\begin{proof}
Let $\xi\in [l_n-\epsilon, l_n+\epsilon]$. We once again use the method of variation of parameters but this time about the solutions $u_p(-, l_n+\epsilon)$ and $y_p(-, l_n+\epsilon)$, i. e. the periodic solutions as before but at $\xi = l_n+\epsilon$ fixed. Like in the previous lemma, $u_p(x, l_n+\epsilon)$, $y_p(x, l_n+\epsilon) < K$ where $K$ is constant in $x$ and $\xi\in \{l_n, r_n\}_{n\in\mathbb{N}}\cap\left[0,\frac{4M^2}{\epsilon^2}\right]$. We get
\begin{align*}
\left|\left(\begin{array}{c}
a(x)\\
b(x)\end{array}\right)\right| &= \left|\left(\begin{array}{c}
1\\
0\end{array}\right)+\int_0^x\frac{l_n +\epsilon-\xi+q(x)}{d}\left(\begin{array}{cc}
-y_p u_p & -y_p^2\\
u_p^2 & u_py_p \end{array}\right)\left(\begin{array}{c}a(t) \\ b(t)\end{array}\right)dt\right|\\
&\leq 1 + K_1\int_0^x(2\epsilon+|q(x)|)\left|\left(\begin{array}{c}a(t) \\ b(t)\end{array}\right)\right|dt.
\end{align*}
As in proof of the previous lemma, applying Gronwall inequality and picking $C$ appropriately we get (\ref{regularity bounds}).
\end{proof}

The three lemmas imply Proposition \ref{regularity for perturbed periodic}.
From Lemma \ref{large xi} we get (\ref{regularity bounds}) for large $\xi$.
This leaves only finitely many bands, so it suffices to consider
the remaining bands one at a time as in Lemmas \ref{near endpoints} and \ref{interior}.
\section{Variational Principle and the Christoffel-Darboux Formula}\label{variation and cd formula}

We let $T_LF(\xi)=\int F(\zeta )S_L(\xi, \zeta)d\mu(\zeta)$, where $d\mu = d(s\nu)$ is a scalar multiple of a spectral measure $d\nu$. We show that $T_L$ is the orthogonal projection onto $H_L$. We first show

\begin{lemma}\label{fix cos}
The function $\cos(\sqrt{\xi} N)$ is fixed by $T_L$ for $N\leq L$.
\end{lemma}

\begin{proof}
Let $u$ be the solution associated to $d\mu$. There exists a continuous integration kernel $M$ (\cite{gelfand-levitan}, (1.2.5'') \cite{marchenko}) such that
\begin{equation}\label{marchenko 1.2.5''}
\frac{\cos(\sqrt{\xi} x)}{\sqrt{s}} = u(\xi, x) + \int_0^x M(x, t)u(\xi, t)dt.
\end{equation}
Substituting this expression for $\frac{\cos(\sqrt{\xi} x)}{\sqrt{s}}$ in evaluating $T_L(\frac{\cos(\sqrt{\xi} x)}{\sqrt{s}})$, we check
\begin{align*}
\int \frac{\cos(\sqrt{\xi}  N)}{\sqrt{s}}S_L(\zeta , \xi )d\mu(\xi)&= u(\xi , N) + \int_0^N M(N, t)\int u(\xi , t)S_L(\zeta , \xi )d\mu(\xi)dt = \\
= u(\xi , N) &+ \int_0^N M(N, t)u(\xi , t)dt = \frac{\cos(\sqrt{\xi} N)}{\sqrt{s}}.
\end{align*}
Here we use Fubini's theorem, the reproducing property of $S_L$ (noting that $N\leq L$), and we recover the last equality again by (\ref{marchenko 1.2.5''}).
\end{proof}

We then show that $T_L$ fixes $\pi_N\in H_N$ for $N \leq L$.
\begin{corollary}\label{fix stuff in H_L}
If $\pi_N(\xi) = \int_0^N f(x)\cos(\sqrt{\xi} x)dx$ for some function $f\in L_2([0, N])$ and $N \leq L$, then $\pi_N(\xi) = \int\pi_N(\zeta )S_L(\xi, \zeta ) d\mu(\zeta)$.
\end{corollary}

\begin{proof}
This is a straightforward calculation, using (\ref{marchenko 1.2.5''}):
$$\int\pi_N(\zeta )S_L(\xi, \zeta )d\mu(\zeta) = \int \int_0^N f(x)\cos(\sqrt{\zeta} x)S_L(\xi, \zeta ) dx d\mu(\zeta)= $$ $$=
\int_0^N f(x) \int \cos(\sqrt{\zeta} x)S_L(\xi, \zeta ) d\mu(\zeta)dx = \pi_N(\xi).$$
Here we make use of Fubini's theorem and the Lemma \ref{fix cos}.
\end{proof}

\begin{theorem}\label{projection}
The operator $(T_L\pi_N)(\xi) = \int  \pi_N(\zeta )S_L(\xi, \zeta )d\mu(\zeta)$ is an orthogonal projection onto the Hilbert space $H_L$.
\end{theorem}

\begin{proof}
To show that $T_L$ is a projection, by Corollary \ref{fix stuff in H_L}, it suffices to show that $T_L\pi_N(\xi) \in H_L$ for $N \geq L$. Recalling that $\pi_N(\xi) = \int_0^N f(x)\cos(\sqrt{\xi} x)dx$, we compute:
\begin{align*}
\int \pi_N(\zeta )S_L(\xi, \zeta )d\mu(\zeta) &= \int d\mu(\zeta) \left(\int_0^L + \int_L^N\right) f(x)\cos(\sqrt{\zeta} x)S_L(\xi, \zeta )dx \\
= \pi_L(\xi) &+ \int d\mu(\zeta) \int_L^Nf(x)\cos(\sqrt{\zeta} x)S_L(\xi, \zeta )dx
\end{align*}

We substitute (\ref{marchenko 1.2.5''}) for $\cos(\sqrt{\zeta}  x)$ to get
\begin{align*}
\int d\mu(\zeta) &\int_L^N f(x)\cos(\sqrt{\zeta} x)S_L(\xi, \zeta )dx =\\
&= \int d\mu(\zeta)\int_L^Nf(x)\left(u(\zeta , x) +\int_0^x M(x, t)u(\zeta , t)dt\right)S_L(\xi, \zeta )dx
\end{align*}
By Fubini and the reproducing property of the kernel, the first term is 0. The second term is $\int_0^L g(t) u(\xi, t) dt$, where $g(t) = \int_L^N f(x)M(x, t)dx$, and $\int_0^L g(t)u(\zeta , t)dt\in H_L$ (by (1.2.10) in \cite{marchenko}).

We next check that $T$ is self-adjoint:
$$\langle g, Tf\rangle_{d(s\mu)} = \int d\mu(\xi)g(\xi )\int d\mu(\zeta) f(\zeta )S_L(\xi , \zeta )=$$
$$=\int d\mu(\zeta) f(\zeta ) d\mu(\xi) g(\xi ) S_L(\zeta , \xi ),$$ since our definition of $S_L$ is symmetric in $\zeta$ and $\xi$.

\end{proof}
\bigskip
We now prove Theorem \ref{variational principle}.
\begin{proof}
Fixing $\xi_0\in\mathbb{C}$ we consider
\begin{equation}\label{min}
\inf\{\|\pi\|^2:\pi_L(\xi) = \int_0^L f(x)\cos(\sqrt{\xi} x)dx; \pi(\xi_0) = 1\}.
\end{equation}
If $\phi\neq 0$ is in some Hilbert space $H$, then
\begin{equation}\label{simon min}
\min\{\|\psi\|^2: \langle\psi, \phi\rangle = 1\} = \frac{1}{\|\phi\|^2}
\end{equation}
and the minimizer is given by $\frac{\phi}{\|\phi\|^2}$ (Proposition 1.2.1 of \cite{OPUC1}).
In our case, the Hilbert space is $H_L$. The condition that $\pi(\xi_0) = 1$ is equivalent to $$1 = \pi(\xi_0) = \int d\mu(\zeta) \pi(\zeta ) S_L(\zeta , \xi_0) = \langle\pi, S_L(-, \xi_0)\rangle.$$

The proposition is applicable with $\phi(\xi) = S_L(\xi, \xi_0)\in H_L$ as shown above. Therefore the minimum is equal to $$\frac{1}{\|S_L(-, \xi_0)\|^2} =  S_L(\xi_0, \xi_0)$$
and the minimizer is $$S_L(\xi, \xi_0)/S_L(\xi_0, \xi_0).$$
\end{proof}

We show the analogue of the Christoffel-Darboux formula here:
\begin{lemma}
\begin{equation}\label{Christoffel-Darboux}
S_L(\alpha, \beta) = \frac{u(\alpha, L)u'(\beta, L)-u(\beta, L)u'(\alpha, L)}{\alpha - \beta}
\end{equation}
\end{lemma}

\begin{proof}
\begin{align*}
u(\alpha, x)u''(\beta, x) &= u(\alpha, x)(q(x)-\beta)u(\beta, x)\\
u(\beta, x)u''(\alpha, x) &= u(\beta, x)(q(x) - \alpha)u(\alpha, x)
\end{align*}
We subtract to get
\begin{equation}u(\alpha, x)u''(\beta, x) - u(\beta, x)u''(\alpha, x) = (\alpha - \beta)u(\alpha, x)u(\beta, x)
\end{equation}
Integrating both sides $dx$ from 0 to $L$ we get the desired formula. The left hand side has to be integrated by parts:
\begin{align*}
\int_0^Lu(\alpha, x)u''(\beta, x) &- u(\beta, x)u''(\alpha, x)dx\\
&= u(\alpha, 0)u'(\beta, 0)-u(\alpha, L)u'(\beta, L) - u(\beta, 0)u'(\alpha, 0)+u(\beta, L)u'(\alpha, L)\\
&= u(\beta, L)u'(\alpha, L)-u(\alpha, L)u'(\beta, L),
\end{align*}
for any boundary condition given at 0 and independent of $\alpha, \beta$, such as Dirichlet or Neumann.
\end{proof}

On the diagonal, the Christoffel-Darboux formula becomes
\begin{equation}\label{Christoffel-Darboux diagonal}
S_L(\xi, \xi) = u'(\xi, x)\frac{d}{d\xi}u(\xi, x) -\frac{d}{d\xi}u'(\xi, x)u(\xi, x)
\end{equation}

\section{Bounds on the Diagonal Kernel}\label{section bound}

We will show the analogue of Lemma 3.1 in Simon \cite{simon2ext}. Assume regularity bounds (\ref{regularity bounds}) on the measure $d\mu$. Let
\begin{equation}
Q_L(\xi, \xi_0) = \frac{S_L(\xi, \xi_0)}{S_L(\xi_0, \xi_0)}.
\end{equation}
be the minimizer in (\ref{min}), then

\begin{lemma}
Let $d\mu$ be  a measure that satisfies regularity bounds. Then for all $\epsilon>0$ there exist $C$, $\delta_1$ such that
$|Q_L(\xi)| \leq Ce^{\epsilon L}\lambda_L(\xi_0)$, for $\xi \in \{\xi:\dist(\xi, \mathfrak{e})\leq\delta_1\}$
\end{lemma}

\begin{proof}
Fix $\epsilon$. A regularity bound (\ref{regularity bounds}) on a measure $d\mu$ implies a bound on $|S_L(\xi, \xi_0)|$ by Cauchy-Schwarz:
\begin{align*}
|S_L(\xi, \xi_0)| &\leq \left(\int_0^L u(\xi, x)^2dx\right)^{1/2}\left(\int_0^Lu(\xi_0, x)^2
\right)^{1/2}\leq Ce^{\epsilon L}.
\end{align*}
\end{proof}

To show Lemma \ref{lemma bound} we need the following fact about the spectral measure:
\begin{lemma}\label{spectral measure}
Let $A$ be a self adjoint Schr\"odinger operator and $d\mu$ be a scalar multiple of its spectral measure. Then for $n\geq 2$ there exists a constant $K$
\begin{equation}\label{mu at infinity}
\int_2^{\infty}\frac{d\mu(\xi)}{\xi^n}\leq K2^{-n}.
\end{equation}
\end{lemma}

\begin{proof}
This follows easily from Marchenko (Theorem 2.4.2 of  \cite{marchenko}) for the Neumann boundary condition and Section 6 of Gesztesy-Simon \cite{gesztesy-simon} for Dirichlet boundary condition.
\end{proof}

\begin{lemma}\label{lemma bound}

Suppose $d\mu(\xi) = w(\xi)d\xi + d\mu_s$, $d\mu^{*}(\xi) = w^{*}(\xi)d\xi + d\mu^{*}_s$ are two unnormalized spectral measures with $\sigma_{\ess}(d\mu)=\sigma_{\ess}(d\mu^{*})=\mathfrak{e}$. Suppose $d\mu$, $d\mu^{*}$ satisfy regularity bounds and have finitely many eigenvalues outside of $\{\xi: \dist(\xi, \mathfrak{e}) < \delta_1\}$ for any $\delta_1>0$. Let $I\subset \mathfrak{e}^{int}$ be a closed and bounded interval such that $w, w^{*}$ are continuous and strictly positive on $I$ and $(\supp(d\mu_s)\cup \supp(d\mu^{*}_s))\cap I = \emptyset$. Let $\xi_0\in I$ and $\xi(L) \rightarrow \xi_0$ as $L\rightarrow\infty$. Then for all sufficiently small $\delta$ and all $\epsilon > 0$ and all $M$ there exist $\gamma < 1$, $C$, $n$ such that for all $N > n+1$

\begin{equation}
\lambda_L(\xi_0, \mu^{*}) \leq \sup_{|\xi - \xi_0|<\delta}\left(\frac{w^{*}(\xi)}{w(\xi)}\right)\lambda_M(\xi_0, \mu) + Ce^{2\epsilon M}\gamma^N + Ce^{2\epsilon M}2^{-2N}
\end{equation}\label{bound}

where $L = M + \frac{\pi}{4\xi_0}N$.
\end{lemma}

\begin{proof}
We use the methods of Lubinsky \cite{bulk} and Simon \cite{simon2ext}.

Let $Q_M$ be the minimizing function for the measure $\mu$ and
\begin{equation}\label{F}
F(\xi) = \frac{4\xi_0}{T\pi}\left(\frac{\sin(\frac{\pi}{4\xi_0}(\xi - \xi_0))}{\xi - \xi_0} +  \frac{\sin(\frac{\pi}{4\xi_0}(\xi+\xi_0))}{\xi+\xi_0}\right),
\end{equation}
where $T = 1+\frac{2}{\pi}$.

We notice that
\begin{enumerate}
\item $|F(\xi_0)| = 1$

\item $|F(\xi)| < \gamma$ whenever $|\xi - \xi_0|\geq\delta$, for some $0<\gamma <1$ depending on $\delta$, and

\item $|F(\xi)|<\frac{C\xi_0}{|\xi -\xi_0|}$ whenever $|\xi -\xi_0| > 1$.

\end{enumerate}

The function $F$ is just $\frac{\sin(\xi)}{\xi}$ shifted so that 0 is at $\xi_0$, scaled so that exactly one period of the sine happens between 0 and $\xi_0$, then symmetrized to make it even, and then scaled by a factor of $\frac{1}{T}$ again to make $F(\xi_0) = 1$. Since $\frac{\sin\xi}{\xi} = \int_0^1 \cos(\xi x)$, $F$ is a Fourier transform of some even function $f$ supported on $[-\frac{\pi}{4\xi_0}, \frac{\pi}{4\xi_0}]$ and $F^N$ is the Fourier transform of an even function with support in $[-\frac{N\pi}{4\xi_0}, \frac{N\pi}{4\xi_0}]$.

Fix $\epsilon$. Since the measures $d\mu$ and $d\mu^{*}$ are essentially supported on the same set $\mathfrak{e}$, we can let $\delta_1$ as in the definition of regularity bounds (\ref{regularity bounds}) for both measures. Let $\mathfrak{e}_{\delta_1} = \{\xi: \dist(\xi, \mathfrak{e}) < \delta_1\}$. We label the mass points of $d\mu^{*}$ outside $\mathfrak{e}_{\delta_1}$ with $\{\xi_1, \xi_2, \xi_3, ..., \xi_n\}$. We can construct a polynomial $P$ with zeros at $\xi_1, ..., \xi_n$ and a local maximum at $\xi_0$ of $P(\xi_0)=1$ with degree $n+1$.

Then let
$$Q(\xi) = Q_M(\xi, \xi_0, \mu)F^NP.$$

Since $Q(\xi_0) = 1$, by the minimizing property of $\lambda_L$,
$$\|Q\|^2_{H_L(d\mu^{*})}\geq \lambda_L (\xi_0, \mu^{*}).$$
We then find a bound on $\|Q\|^2_{H_L(d\mu^{*})}$ from above.
$$\|Q\|^2 = \int |Q(\xi )|^2d\mu^{*}(\xi) = (\int_{|\xi - \xi_0|<\delta} + \int_{|\xi - \xi_0|\geq\delta})|Q(\xi )|^2d\mu^{*}(\xi),$$

Both $F$ and $P$ have a local maximum of 1 at $\xi_0$, so we see that
\begin{align*}
\int_{|\xi -\xi_0|<\delta}|Q(\xi )|^2 d\mu^{*}(\xi)
&\leq \sup_{|\xi_0 - \xi|<\delta}\frac{w^{*}(\xi)}{w(\xi)}\int_{|\xi_0-\xi|<\delta}|Q_M(\xi )|^2 d\mu(\xi) \\
&\leq \sup_{|\xi_0-\xi|<\delta}\frac{w^{*}(\xi)}{w(\xi)}\lambda_M(\xi_0, \mu).
\end{align*}

The measure $\mu^{*}$ is pure point on $\mathbb{R}\backslash \mathfrak{e}_{\delta_1}$ and the zeros of $P$ coincide with the mass points of $\mu^{*}$, so integrating $|F^NP|^2$ over the set $\mathfrak{e}_{\delta_1}$ is the same as the integrating over $\mathbb{R}$. We use (\ref{regularity bounds}) to show that the integral of $|Q^2|$ over $|\xi -\xi_0|\geq\delta$ is small for large $N$:

\begin{align*}
\int_{|\xi - \xi_0|\geq\delta} &|Q(\xi )|^2 d\mu^{*}(\xi)  \leq \frac{C\lambda_M(\xi_0)e^{4\epsilon M}}{T}
\int_{|\xi - \xi_0|\geq\delta, \xi \in \mathfrak{e}_{\delta_1}}| F(\xi )^NP(\xi )|^2d\mu^{*}(\xi)\\
\leq&\frac{C\lambda_M(\xi_0)e^{4\epsilon M}}{T}\left(\int_{\delta \leq |\xi -\xi_0| \leq2} + \int_{|\xi -\xi_0| >2}\right)|F(\xi )|^{2N}P^2(\xi)d\mu^{*}(\xi).
\end{align*}

We have split the integral into two pieces: one that is close to $\xi_0$ and one that is far. For the close piece, since 1 is a maximum of $F$ on $[\xi_0 - 2, \xi_0+2]$ there exists $\gamma<1$ such that $F(\xi) < \gamma$ on $\{\xi:\delta <|\xi  - \xi_0|\leq 2\}.$ Therefore,
\begin{align*}
\int_{\{\xi:|\xi  - \xi_0|\leq 2\}\backslash[\xi_0 - \delta, \xi_0+\delta]}|F(\xi )|^{2N}P^2(\xi)d\mu^{*}(\xi)
&\leq C\gamma^{2N}.
\end{align*}

For the second piece,
$$\int_{|\xi  - \xi_0|> 2}|F(\xi )|^{2N}P^2(\xi)d\mu^{*}(\xi) \leq
\int_{|\xi  - \xi_0 |> 2}\frac{C\xi^{2n+2}\xi_0}{(\xi  - \xi_0 )^{2N}}d\mu^{*}(\xi) \leq C\xi_02^{-2N},$$
for $N > n+1$. The last bound follows from Lemma \ref{mu at infinity}.

Since $\xi_0\in I\subset \mathfrak{e}^{int}$ for a compact interval $I$ and $\lambda_M(\xi_0)$ is continuous on $I$, we can choose $C$ that is uniform in $\xi_0$ on $I$ in Lemma \ref{bound}.
\end{proof}

We now prove Theorem \ref{diagonal kernel}

Suppose $d\mu^{*}$, $d\mu$, $I$ as in theorem and let $\xi(L) \rightarrow \xi_0 \in I$.

Fix $\delta$, $\epsilon$. Let $\delta_1$ be small enough so that regularity bounds (\ref{regularity bounds}) hold for both $\mu$, $\mu^{*}$ on $E_{\delta_1}$ and let $n$ be the number of mass points of $\mu^*$ outside of $E_{\delta_1}$. Pick $N_1, N_2 > (n+1)/\epsilon$ so that $(1/2)^{N_1}< e^{-4}$ and $\gamma^{N_2} < e^{-4}$. Let $N_3 = \max\{N_1, N_2\}$ and $N =2N_3M\epsilon$, so that Lemma \ref{lemma bound} is applicable and the sum of the second and third terms in (\ref{bound}) is $O(e^{-\epsilon M})$. Divide by $\lambda_L(\xi_0 , \mu)$ to get
\begin{equation}\label{divide}
\frac{\lambda_L(\xi_0 , \mu^{*})}{\lambda_L(\xi_0 , \mu)}\leq \sup_{|\xi -\xi_0|<\delta}\left(\frac{w^{*}(\xi)}{w(\xi)}\right)\frac{\lambda_M(\xi_0 , \mu)}{\lambda_L(\xi_0 , \mu)}+O(e^{-2\epsilon M})S_L(\xi_0 , \xi_0 , \mu).\end{equation}
From regularity bounds (\ref{regularity bounds}) on $\mu$ and for fixed $N$, the second term on the right hand side tends to 0 as $M\rightarrow\infty$:
$$O(e^{-2\epsilon M})S_L(\xi_0 , \xi_0 , \mu)\leq O(e^{-2\epsilon M})Ce^{\epsilon(M + \frac{\pi}{4\xi_0 }N)}= O(e^{-\epsilon M}).$$

Then we take $\inf_{|\xi  - \xi_0 |<\delta}$ on both sides of (\ref{divide}) and we adjust the sup accordingly to get
$$\inf_{|\xi -\xi_0 |<\delta}\frac{\lambda_L(\xi , \mu^{*})}{\lambda_L(\xi , \mu)}\leq \sup_{|\xi -\xi_0|<2 \delta}\left(\frac{w^{*}(\xi)}{w(\xi)}\right)\inf_{|\xi -\xi_0 |<\delta}\frac{\lambda_M(\xi , \mu)}{\lambda_L(\xi , \mu)}.$$
We then let $\delta \rightarrow 0$, then $M\rightarrow\infty$, and then $\epsilon\rightarrow 0$, we get by continuity and positivity of $w$ that
$$\liminf_{L\rightarrow\infty} \frac{\lambda_L(\xi(L), \mu^{*})}{\lambda_L(\xi(L), \mu)} \leq \frac{w^{*}(\xi_0)}{w(\xi_0)}.$$

To get the opposite inequality, we can interchange $\mu$ and $\mu^{*}$ in (\ref{bound}), use the corresponding $N$ given by the same formula, and divide by $\lambda_L(\xi_0 , \mu^{*})$.

All arguments given are uniform in $\xi_0 \in I$.

\section{Calculation of the reproducing kernel in the case of a periodic potential}\label{model}

As in Gesztesy--Zinchenko ((2.8) of \cite{gesztesy-zinchenko}), for $z\in \mathbb{C}\backslash \mathbb{R}$ let $\psi(z, -)\in L^2,$ with  $\psi(z, 0) = 1$. Then the $m$-function is given by
\begin{equation}\label{psi}
\psi(z, x) = y(z, x) - m(z)u(z, x).\end{equation}
Similarly let $\tilde{\psi}$ be the $L^2$ solution with $\tilde{\psi}'(z, 0) = 1$. Then the corresponding $m$-function is given by
\begin{equation}
\tilde{\psi}(z, x)=u(z, x)+\tilde{m}y(z, x).
\end{equation}

\begin{theorem}\label{theorem kernel calculation}
Let $A^{\#} = -\frac{d^2}{dx^2}+p$ be a Schr\"odinger operator with continuous periodic potential $p$ and either the Neumann or the Dirichlet boundary condition, and let  $\rho(\xi)d\xi$ be its density of states. Let $\xi_0\in I\subset \sigma_{\ess}(A^{\#})^{int}$, where $I$ is a closed and bounded interval. Then for $a, b\in \mathbb{R}$ uniformly in $I$
\begin{enumerate}
\item\label{item diagonal}
\begin{equation}\label{diagonal kernel calculation}
\lim_{L\rightarrow\infty}\frac{S_L(\xi_0, \xi_0)}{\pi L} = \frac{\rho(\xi_0)}{w(\xi_0)}
\end{equation}
and

\item\label{item off-diagonal}
\begin{equation}\label{kernel}
\frac{S_L(\xi_0 + \frac{a}{L}, \xi_0+\frac{b}{L})}{S_L(\xi_0, \xi_0)} = \frac{\sin(\pi\rho(\xi_0)(b-a))}{\pi\rho(\xi_0)(b-a)}.
\end{equation}

\item
Furthermore, (\ref{model condition 5}) is satisfied.

\end{enumerate}

\end{theorem}

\begin{proof}
The methods used here are similar to \cite{simon2ext}.

(\ref{item diagonal})
We first show convergence then uniformity. We use the well known formula relating the $\rho(\xi)$ and $\Im G$, where $G$ is the Green's function. Gesztesy--Zinchenko ((2.18) of  \cite{gesztesy-zinchenko}) gives the Green's function explicitly, so we compute:
\begin{align*}
\rho(\xi) &=\lim_{L\rightarrow\infty} \frac{1}{L}\lim_{\epsilon\downarrow 0}\int_0^L\Im(G(x, x, \xi + i\epsilon))dx\\
&=\lim_{L\rightarrow\infty}\frac{1}{L}\lim_{\epsilon\downarrow 0}\int_0^L\Im(u(\xi + i\epsilon, x)\psi(\xi + i\epsilon, x))dx\\
&=\lim_{L\rightarrow\infty}\frac{1}{L}\lim_{\epsilon\downarrow 0}\Im m(\xi + i\epsilon)\int_0^Lu(\xi, x)^2dx\\
&=\lim_{L\rightarrow\infty}\frac{w(\xi)}{\pi L}\int_0^Lu(\xi, x)^2dx
\end{align*}
Now, $\lim_{\epsilon\downarrow 0}\Im m(\xi + i\epsilon) = w(\xi)$ a.e., so the equality holds a.e..

We use continuity to show equality everywhere and uniformity of convergence. We let $\xi\in I\subset \mathfrak{e}^{int}$ and $f(\xi, x)= e^{i\theta(\xi)x}\phi(\xi, x)$ be the Floquet solution normalized so that $f'(\xi, 0) = 1$. Here $\phi$ is periodic in $x$ as in \cite{magnus-winkler}. Then $f(\xi, 0)\notin \mathbb{R}$, and we claim that
\begin{equation}\label{u in terms of f}
u(\xi, x) = \frac{f(\xi, x)-\overline{f(\xi, x)}}{f(\xi, 0)-\overline{f(\xi, 0)}}
\end{equation}
Since $f$, $\overline{f}$ are solutions of the eigenvalue equation, so is the right hand side of (\ref{u in terms of f}). Therefore it suffices to check that the right hand side satisfies the Neumann boundary conditions, and it does.

Let
\begin{equation}
g(\xi, x) = \frac{\phi(\xi, x)}{f(\xi,0) - \overline{f(\xi, 0)}}.
\end{equation} Then
\begin{equation}\label{u in terms of g}
u(\xi, x) = e^{i\theta(\xi)x}g(\xi, x) + e^{-i\theta(\xi)x}\overline{g(\xi, x)}.
\end{equation}
The Wronskian of $e^{i\theta(\xi)x}g(\xi, x)$ and $e^{-i\theta(\xi)x}\overline{g(\xi, x)}$ is
\begin{align*}
W(\xi) &= -2ig(\xi, x)\overline{g(\xi, x)}\theta(\xi)-g(\xi, x)\overline{g'(\xi, x)}+\overline{g(\xi, x)}g'(\xi, x)
\end{align*}

Substituting (\ref{u in terms of g}) for $u$ in (\ref{Christoffel-Darboux diagonal}) we get that
\begin{equation}
S_L(\xi, \xi)= 2\theta'(\xi)iLW(\xi) + O(1).
\end{equation}
where $O(1)$ is bounded uniformly in $\xi\in I$ and $L$. Both $2\theta'(\xi)iW(\xi)$ and $\frac{\pi\rho(\xi)}{w(\xi)}$ are continuous in $\xi$ and equal a.e., meaning that
\begin{equation}\label{Ww}
\lim_{L\rightarrow\infty}\frac{S_L(\xi, \xi)}{L} =2\theta'(\xi)iW(\xi) = \frac{\pi\rho(\xi)}{w(\xi)}
\end{equation}
for all $\xi\in I$. The convergence in (\ref{diagonal kernel calculation}) is uniform.

A similar argument yields the result for $S_L$ corresponding to the Dirichlet boundary condition.

\bigskip
(\ref{item off-diagonal})
For the Floquet solution $f$ normalized so that $f'(\xi, 0) = 1$ we have
\begin{align*}
f(\xi, Pk+s) &= f(\xi, s)e^{ik\theta(\xi)}\\
\end{align*}
  By analytic perturbation theory (e. g. Theorems XII.13 and XII.3 of \cite{reed-simon4}), $f$ is real analytic in $\theta$ for $\theta\in (0, \pi)\cup (\pi, 2\pi)$ and at closed gaps i.e. $\theta = \pi$ and $\Delta'(\theta)=0$. By Theorem XIII.89 of \cite{reed-simon4}, $\xi(\theta)$ is analytic and $\xi'(\theta)\neq 0$, which implies that $\theta(\xi)$ is analytic on the interiors of the bands. The function $\theta(\xi)$ is also analytic at $\xi_0$ if $\xi_0$ is a closed gap. To see this we take the derivative of the discriminant equation $\Delta(\xi) = 2\cos(\theta)$:
  $$\frac{d}{d\xi}(\Delta(\xi)))=\frac{d}{d\xi}D(\xi)\frac{d}{d\theta}\xi(\theta)=-2\sin(\theta).$$
  At a closed gap $\xi_0$, the right hand side has a single zero and $\frac{d}{d\xi}D(\xi)$ also has a single zero. This implies that $\frac{d}{d\theta}\xi(\theta)\neq 0$ at a closed gap so that $\theta(\xi)$ is analytic at $\xi_0$.

  We can therefore take the Taylor series of $\theta(\xi)$, $f(\xi, s)$, and $f'(\xi, s)$ to get
\begin{align}
f(\xi_0 + \frac{a}{L}, x) &= (f(\xi_0, s) + O(\frac{1}{L}))e^{ik(\theta(\xi_0) + \frac{a\theta'(\xi_0)}{L}+O(\frac{1}{L^2}))}\\
\frac{d}{dx}f(\xi_0 + \frac{a}{L}, x) &= (\frac{d}{ds}f(\xi_0, s) + O(\frac{1}{L}))e^{ik(\theta(\xi_0) + \frac{a\theta'(\xi_0)}{L}+O(\frac{1}{L^2}))}
\end{align}
Letting $L = Pk +s$, we substitute this into (\ref{u in terms of f}) to get
\begin{align*}
&2y(\xi_0 + \frac{a}{L}, x)\Im f(\xi_0 + \frac{a}{L}, 0) = \\
&(f(\xi_0, s) + O(\frac{1}{L}))e^{ik(\theta(\xi_0) + \frac{a\theta'(\xi_0)}{L}+O(\frac{1}{L^2}))}-\overline{(f(\xi_0, s) + O(\frac{1}{L}))}e^{-ik(\theta(\xi_0) + \frac{a\theta'(\xi_0)}{L}+O(\frac{1}{L^2}))}
\end{align*}

From this we get the following by direct computation:
\begin{align*}
&2\Im f(\xi_0 + \frac{a}{L}, 0)\Im f(\xi_0 + \frac{b}{L}, 0)(y(\xi_0+\frac{a}{L}, L)y'(\xi_0+\frac{b}{L}, L)-y(\xi_0+\frac{b}{L}, L)y'(\xi+\frac{a}{L}, L))\\
&= W(f, \overline{f})i\sin(\frac{a-b}{P}\theta'(\xi_0) + O(L^{-1})).
\end{align*}
Then substituting into the left hand side of (\ref{kernel}), we get
\begin{align*}
&\frac{S_L(\xi_0 + \frac{a}{L}, \xi_0+\frac{b}{L})}{S_L(\xi_0, \xi_0)}=\\
 &= \lim_{L\rightarrow\infty}\frac{w(\xi_0)\Im(f(\xi_0, 0))(y(L, \xi_0+\frac{a}{L})y'(L, \xi_0+\frac{b}{L})-y(L, \xi_0+\frac{b}{L})y'(L, \xi+\frac{a}{L}))}{\Im(f(\xi_0+\frac{a}{L},0))\Im(f(\xi_0+\frac{b}{L}, 0))\rho(\xi_0)(b-a)}\\ &=\frac{\sin(\pi\rho(\xi_0)(b-a))}{\pi\rho(\xi_0)(b-a)}
\end{align*}
Here we have used that
\begin{equation}\label{w = f}
w(\xi)= \Im f(\xi, 0),
\end{equation}
which we get by substituting
\begin{equation}
W(\xi) = \frac{f(0)\overline{f'(0)}-f'(0)\overline{f(0)}}{(2i\Im f(\xi, 0))^2} = (2i\Im f(\xi, 0))^{-1},
\end{equation}
in (\ref{Ww}).

An identical calculation yields the result for the Dirichlet boundary condition.

\bigskip
To show (\ref{model condition 5}), let $\epsilon(L)\rightarrow 0$ as $L\rightarrow\infty$. Since $u$ is real analytic in $\xi$,

$$u^2(\xi + \epsilon(L), x) = u^2(\xi, x) + \frac{d}{d\xi}(u^2(\xi, x))\epsilon(L) + o(\epsilon(L)),$$
and since $I$ is compact, $\frac{d}{d\xi}(u^2(\xi, x))$ achieves a maximum, so that $u^2(\xi + \epsilon(L), x) = u^2(\xi, x) + O(\epsilon(L))$ uniformly on $I$. Thus,
\begin{align*}
\lim_{L\rightarrow\infty}\frac{w(\xi)}{\pi L}&\int_0^L u(\xi + \epsilon(L), x)^2dx =\\
&= \lim_{L\rightarrow\infty}\frac{w(\xi)}{\pi L}\int_0^Lu(\xi, x)^2dx + O(\epsilon(L)).
\end{align*}

\end{proof}

\section{Off-Diagonal Kernel and Clock Behavior}\label{section off-diagonal kernel}

The main goal of this section is to prove our main result Theorem \ref{off-diagonal kernel}. We start by proving Lubinsky's inequality, which is similar to the discrete case:
\begin{lemma}
Let two measures $d\mu(\xi), d\mu^{*}(\xi)$ with $d\mu(\xi) \leq d\mu^{*}(\xi)$ be unnormalized spectral measures of Schr\"odinger operators. Then for any $\xi$, $\beta\in\mathbb{R}$,
\begin{equation}
\frac{|S_L(\xi, \beta, \mu) - S_L(\xi, \beta, \mu^{*})|}{S_L(\xi, \xi, \mu)}\leq \left(\frac{S_L(\beta, \beta, \mu)}{S_L(\xi, \xi, \mu)}\right)^{1/2}\left(1-\frac{S_L(\xi, \xi, \mu^{*})}{S_L(\xi, \xi, \mu)}\right)^{1/2}.
\end{equation}
\end{lemma}

\begin{proof}
The proof carries over from \cite{bulk}.
Expanding,
\begin{align*}
&\int(S_L(\xi, \zeta, \mu) - S_L(\xi, \zeta, \mu^*))^2d\mu(\zeta) =\\
& = \int S_L(\xi, \zeta, \mu)^2d\mu(\zeta) - 2\int S_L(\xi, \zeta, \mu)S_L(\xi, \zeta, \mu^*)d\mu(\zeta) + \int S_L^{2}(\xi, \zeta, \mu^*)d\mu(\zeta)\\
&= S_L(\xi, \xi, \mu)-2S_L(\xi, \xi, \mu^*) + \int S_L(\xi, \zeta, \mu^*)d\mu(\zeta).
\end{align*}

Since $d\mu \leq d\mu^*$,
\begin{equation}
\int S_L(\xi, \zeta, \mu^*)d\mu(\zeta)\leq \int S^{2}(\xi, \zeta, \mu^*)d\mu^*(\zeta) = S_L^*(\xi, \xi)
\end{equation}

Therefore,
\begin{align*}
\int (S_L(\xi, \zeta, \mu) - S_L(\xi, \zeta, \mu^*))^2d\mu(\zeta) \leq S_L(\xi, \xi, \mu)- S_L(\xi, \xi, \mu^*).
\end{align*}

Using the variational principle for the Christoffel--Darboux symbol e.g. the minimizing property, for any $\pi(\zeta)\in H_L$ and any $\beta\in \mathbb{R}$
$$S_L(\beta, \beta, \mu)^{-1} \leq \int \frac{\pi(\zeta)^2}{\pi(\beta)^2}d\mu(\zeta).$$
Using $\pi(\zeta) = S_L(\xi, \zeta, \mu) - S_L(\xi, \zeta, \mu^*)$ we get that
$$|S_L(\xi, \beta, \mu)-S_L(\xi, \beta, \mu^*)|\leq S_L(\beta, \beta, \mu)^{1/2}(S_L(\xi, \xi, \mu^*)-S_L(\xi, \xi, \mu^*))$$
\end{proof}

We then show
\begin{lemma}\label{staggering}
Let $d\mu$, $d\mu^*$ be unnormalized spectral measures with $\sigma_{\ess}(d\mu) = \sigma_{\ess}(d\mu^{*})$. If $d\mu(\xi)$ obeys regularity bounds and $d\mu(\xi) \leq d\mu^{*}(\xi)$ then $d\mu^{*}(\xi)$ also obeys regularity bounds.
\end{lemma}

\begin{proof}
Since $d\mu \leq d\mu^{*}$, $\|Q\|_{d\mu} \leq \|Q\|_{d\mu^{*}}$ for all $Q\in L^2(d\mu)\cap L^2(d\mu^{*})$, so \begin{align*}
&\inf\{\|Q\|_{d\mu}:Q(\xi_0) =1, Q(\xi)=\int_0^L f(x)\cos(\sqrt{\xi} x)dx\} \\
\leq &\inf\{\|Q\|_{d\mu^{*}}:Q(\xi_0) =1, Q(\xi)=\int_0^L f(x)\cos(\sqrt{\xi} x)dx\}.
\end{align*}
By the variational principle, this implies that $\lambda_L(\xi, \mu) \leq \lambda_L(\xi, \mu^{*})$. If $u$, $u^*$ are the solutions of the eigenvalue equations corresponding to $d\mu$, $d\mu^{*}$ respectively, then $$Ce^{\epsilon L}\geq\int_0^L u(\xi, x)^2dx \geq \int_0^L u^*(\xi, x)^2dx.$$
\end{proof}

We now prove Theorem \ref{off-diagonal kernel}.

\begin{proof}
Let $A = -\frac{d^2}{dx^2} + p(x) + q(x)$ and $A^{\#} = -\frac{d^2}{dx^2} + p(x)$ be Schr\"odinger operators with periodic continuous $p$ and non-destructive zero-average $q$ (Definition \ref{non-destructive zero-average}). Suppose the corresponding spectral measures $d\mu$, $d\mu^{\#}$ satisfy regularity bounds. Suppose there exists a closed and bounded interval $I\subset \sigma_{\ess}(A)^{int}$ such that $\xi_0\in I$, $w$ is absolutely continuous and positive on $I$, and $(\sigma_{\ess}(d\mu_s)\cup \sigma_{\ess}(d\mu^{\#}_s))\cap I = \emptyset.$

Let $s>0$ such that $sw^{\#}(\xi_0) = w(\xi_0)$. From $\mu$, $\mu^{\#}$ we construct a new unnormalized spectral measure $\mu^{*}$ which dominates $\mu$, $s\mu^{\#}$ and is absolutely continuous on $I$ with $w^{*}(\xi_0)=w(\xi_0)$. Let $d\mu^{*}(\xi) = \sup\{sd\mu^{\#}(\xi), d\mu(\xi)\}$, for $\xi < R$ and $d\mu^{*}(\xi) = sd\mu^{\#}(\xi) + d\mu(\xi)$ for $\xi \geq R$ where $R\in \mathbb{R}$ with $I\subset (-\infty, R)$. We claim that $\mu^*$ is an unnormalized spectral measure.

A measure $d\nu$ is a spectral measure for a boundary value problem (Theorem 2.3.1 of \cite{marchenko}) if and only if
\begin{enumerate}
\item
The functional on $H_L$ given by the inner product $\langle-, \pi(\xi)\rangle_{d\nu}$ is non-trivial for all non-trivial $\pi$.
\item
The function
\begin{equation}
\Phi(x, \nu) = \int \frac{1-\cos(\sqrt{\xi}x)}{\xi}d\nu(\xi)
\end{equation}
is thrice continuously differentiable in $x$ and $\Phi'(0+, \nu) = 1$.
\end{enumerate}
Condition (1) is true for $d\mu^{*}$, since it is true for both $\mu$ and $\mu^{\#}$. To show condition (2), let $\Phi_R(x, \nu) =\int_{-\infty}^R \frac{1-\cos(\sqrt{\xi}x)}{\xi}d\nu$, for any locally finite measure $d\nu$. Then $\Phi_R(x, \mu)$, $\Phi_R(x, \mu^{\#})$, $\Phi_R(x, \mu^{*})$ are in $C^{\infty}$ by Dominated Convergence Theorem and $$
\int_R^{\infty}\frac{1-\cos(\sqrt{\xi}x)}{\xi}d\mu^{*} = \Phi(x, \mu)-\Phi_R(x, \mu) + \Phi(x, \mu^{\#}) - \Phi_R(x, \mu^{\#})$$

is in $C^3$ as a sum of $C^3$ functions, making $\Phi(x, \mu^{*})\in C^3$. By continuity of $\Phi_R'(x)$ and the Dominated Convergence Theorem $$\Phi_R'(0+, \mu^*) = \Phi_R'(0, \mu^*) = \int_0^R \frac{\sin(0)}{\sqrt{\xi}}d\mu^*(\xi) = 0,$$
so $$\Phi'(0+, \mu^*)=\Phi'(0+, \mu)+\Phi'(0+, \mu^{\#}) = 1+s.$$
Thus, dividing $d\mu^{*}$ by $1+s$ will yield a spectral measure. Additionally, the boundary condition of $d\mu^{*}$ is the same as that for $d\mu$, $d\mu^{\#}$(Theorem 2.4.2 of Marchenko \cite{marchenko}).

By Lemma \ref{staggering} above, $\mu^{*}$ obeys the regularity bound. Thus, by (\ref{diagonal kernel})
$$\frac{S_L(\xi_0+a/L, \xi_0 +a/L, \mu)}{S_L(\xi_0+b/L, \xi_0+b/L, \mu^{*})}\rightarrow 1$$
and
$$\frac{S_L(\xi_0+a/L, \xi_0 +a/L, s\mu^{\#})}{S_L(\xi_0+b/L, \xi_0+b/L, \mu^{*})}\rightarrow 1.$$
Dividing by $S_L(\xi_0, \xi_0)$ and applying Lubinsky's inequality, we get that
$$\frac{|S_L(\xi_0+\frac{a}{L}, \xi_0+\frac{b}{L}, \mu)-S_L(\xi_0+\frac{a}{L}, \xi_0+\frac{b}{L}, \mu^{*})|^2}{S_L(\xi_0 +\frac{b}{L}, \xi_0+\frac{b}{L}, \mu^{*})}$$ $$\leq S_L(\xi_0+\frac{a}{L}, \xi_0+\frac{a}{L}, \mu) -S_L(\xi_0+\frac{a}{L}, \xi_0+\frac{a}{L}, \mu^{*}),$$ and

$$\frac{|S_L(\xi_0+\frac{a}{L}, \xi_0+\frac{b}{L}, s\mu^{\#})-S_L(\xi_0+\frac{a}{L}, \xi_0+\frac{b}{L}, \mu^{*})|^2}{S_L(\xi_0+\frac{b}{L}, \xi_0+\frac{b}{L}, \mu^{*})}$$ $$\leq S_L(\xi_0+\frac{a}{L}, \xi_0+\frac{a}{L}, s\mu^{\#}) -S_L(\xi_0+\frac{a}{L}, \xi_0+\frac{a}{L}, \mu^{*})$$

which gives that
$$\frac{S_L(\xi_0+\frac{a}{L}, \xi_0+\frac{b}{L}, \mu)}{S_L(\xi_0+\frac{a}{L}, \xi_0 +\frac{b}{L}, s\mu^{\#})}\rightarrow 1.$$

Since $$\frac{S_L(\xi_0, \xi_0, \mu)}{S_L(\xi_0, \xi_0, s\mu^{\#})}\rightarrow 1,$$ we get that
$$\lim_{L\rightarrow \infty}\frac{S_L(\xi_0+\frac{a}{L}, \xi_0+\frac{b}{L}, \mu)}{S_L(\xi_0, \xi_0, \mu)} = \lim_{L\rightarrow\infty}\frac{S_L(\xi_0+\frac{a}{L}, \xi_0+\frac{b}{L}, s\mu^{\#}) }{S_L(\xi_0, \xi_0, s\mu^{\#})}.$$

The limit on the right is equal to (\ref{universality limit of the kernel}) and all limits are uniform on $I$ and $|a|, |b| < B$.
\end{proof}

Like \cite{simon2ext}, \cite{levin-lubinsky}, we can now deduce clock spacing of the zeros for a perturbed periodic potential. Here we prove Corollary \ref{clock behavior}.

\begin{proof}
Fix an interval $I\subset \mathfrak{e}^{int}$ and $\xi^*\in I$. We want to show uniform clock behavior at $\xi^*$ of zeros of $u'$ and $y$ in $\xi$ as $L$ gets large. More precisely, if $\xi_n$ is a successive numbering of zeros with $...\xi_{-1} < \xi^* \leq \xi_0 < \xi_1 <...$ then
$$\lim_L L|(\xi_n - \xi_{n+1})|\rho(\xi^*)=1.$$

By the Christoffel--Darboux formula (\ref{Christoffel-Darboux}),  \begin{equation}
\frac{u(\xi^*, L)}{u'(\xi^*, L)} = \frac{u(\xi^* + a/L, L)}{u'(\xi^*+a/L, L)}
\end{equation} for $a \neq 0$ if and only if $S_L(\xi^*, \xi^* + a/L)= 0$. From (\ref{equation diagonal kernel}) and (\ref{diagonal kernel calculation}) we see that $S_L(\xi^*, \xi^*) = O(L)$.
Now, by (\ref{universality limit of the kernel}) and since $S_L(\xi^*, \xi^*) = O(L)$, $S_L(\xi^*, \xi^*+a/L) = o(1/L)$ if and only if $a = \frac{k}{\rho(\xi^*)} +o(1/L)$. The convergence in $L$ is uniform on $I$, since (\ref{universality limit of the kernel}) is uniform on $I$. The argument is the same for $y$.
\end{proof}

\section{Example: the Free Schr\"odinger Operator}\label{free}
The arguments in Section \ref{potential} apply also to non-destructive zero-average perturbations of the free Schr\"odinger operator, thus giving us the regularity bounds condition. We know the spectral measure for the free Schr\"odinger operator \cite{teschl-mathphysics}, and it is indeed continuous and non-negative on $[0, \infty)$.
The solution of the eigenvalue equation for the free Schr\"odinger operator $$-\frac{d^2}{dx^2}u(x, \xi) = \xi u(x, \xi)$$ with the Neumann boundary condition is $\cos(\sqrt{\xi} x)<e^{\epsilon x}$ on $[0, \infty)$. We compute $S_L(\xi, \beta)$ and $S_L(\xi, \xi)$ directly:
$$S_L(\xi, \beta) = \int_0^L \cos(\sqrt{\xi} x)\cos(\sqrt{\beta} x) dx =
\frac{\sin((\sqrt{\xi} - \sqrt{\beta})L)}{2(\sqrt{\xi} - \sqrt{\beta})} + \frac{\sin((\sqrt{\xi}+\sqrt{\beta})L)}{2(\sqrt{\xi} +\sqrt{\beta})},$$
and $$S_L(\xi, \xi) = \frac{L}{2}+\frac{\sin(2\sqrt{\xi} L)}{4\sqrt{\xi}}.$$
Then model property (3) is clear and we check property (4):
$$\limsup_{\epsilon \rightarrow 0}\limsup_{L\rightarrow\infty}\frac{\frac{L+\epsilon L}{2}+\frac{\sin(2\sqrt{\xi} (L+\epsilon L)}{4\sqrt{\xi}}}{\frac{L}{2}+\frac{\sin(2\sqrt{\xi} L)}{4\sqrt{\xi}}} = 1.$$
Locally at $\xi_0$ we get
$$\lim_{L\rightarrow\infty}\frac{S_L(\xi_0 +a/L, \xi_0 + b/L)}{S_L(\xi_0, \xi_0)}= \frac{2\sqrt{\xi_0}\sin(\frac{a - b}{2\sqrt{\xi_0}})}{a - b}$$

This coincides with (\ref{diagonal kernel calculation}), since the density of states for the free Schr\"odinger operator is
\begin{equation}
\rho(\xi) = (2\pi)^{-1}\xi^{-1/2}
\end{equation}
for $\xi \in [0, \infty)$ (Example 8.1 of \cite{berezin-shubin}).

\section{Acknowledgements}
I would like to thank my advisor Professor Barry Simon for all his help. I would also like to thank Professors Jonathan Breuer and Fritz Gesztesy for useful discussions.


\bibliography{universalitylimits}
\end{document}